\documentclass[conference]{IEEEtran}

\usepackage{cite}
\usepackage{amsmath,amsthm,amssymb,amsfonts}
\usepackage{graphicx,wrapfig,lipsum}
\usepackage{siunitx,booktabs}
\usepackage{hyperref}
\usepackage{xspace}
\usepackage{textcomp}
\usepackage{xcolor}
\usepackage{paralist}
\usepackage{algorithmic}
\usepackage[ruled,vlined]{algorithm2e}
\usepackage{pifont,bm}
\usepackage{multirow,footnote}
\makesavenoteenv{table}
\usepackage{thmtools,thm-restate} 
\usepackage{cleveref}[2012/02/15]

\newcommand{\cmark}{\ding{51}}%
\newcommand{\xmark}{\ding{55}}%

\def\BibTeX{{\rm B\kern-.05em{\sc i\kern-.025em b}\kern-.08em
    T\kern-.1667em\lower.7ex\hbox{E}\kern-.125emX}}

\newtheorem{theorem}{Theorem}

\newtheorem{definition}{Definition}

\makeatletter
\newcommand\footnoteref[1]{\protected@xdef\@thefnmark{\ref{#1}}\@footnotemark}
\makeatother

\usepackage{footnote}
\makesavenoteenv{tabular}
\makesavenoteenv{table*}

\usepackage{caption}
\usepackage{tabularx}

\newenvironment{protocol}[1][tb]
  {
   \begin{algorithm}[#1]%
  }{\end{algorithm}}
  
\long\def\com#1{}

\newcommand{\eg}{{e.g.}}
\newcommand{\ie}{{i.e.}}

\newcommand{\sys}{\textsc{Brick}\xspace}
\newcommand{\guard}{{Warden}\xspace}
\newcommand{\guards}{{Wardens}\xspace}
\newcommand{\asys}{\textsc{Brick+}\xspace}

\begin{document}
%
\title{\sys: Asynchronous Payment Channels}

\author{\IEEEauthorblockN{Zeta Avarikioti}
\IEEEauthorblockA{ETH Zurich}
\and
\IEEEauthorblockN{Eleftherios Kokoris-Kogias}
\IEEEauthorblockA{EPFL}
\and
\IEEEauthorblockN{Roger Wattenhofer}
\IEEEauthorblockA{ETH Zurich}
\and
\IEEEauthorblockN{Dionysis Zindros}
\IEEEauthorblockA{NKUA, IOHK}}

\com{
\IEEEoverridecommandlockouts
\makeatletter\def\@IEEEpubidpullup{6.5\baselineskip}\makeatother
\IEEEpubid{\parbox{\columnwidth}{
    Network and Distributed Systems Security (NDSS) Symposium 2021\\
    23-26 February 2020, San Diego, CA, USA\\
    ISBN 1-891562-61-4\\
    https://dx.doi.org/10.14722/ndss.2021.23xxx\\
    www.ndss-symposium.org
}
\hspace{\columnsep}\makebox[\columnwidth]{}}
}

\maketitle

\begin{abstract}
Off-chain protocols (channels) are a promising solution to the scalability and privacy challenges of blockchain payments. 
Current proposals, however,  require synchrony assumptions to preserve the safety of a channel, leaking to an adversary the exact amount of time needed to control the network for a successful attack. 
In this paper, we introduce \sys, the first payment channel that remains secure under network asynchrony and concurrently provides correct incentives.
The core idea is to incorporate the conflict resolution process within the channel by introducing a rational committee of external parties, called \guards. Hence, if a party wants to close a channel unilaterally, it can only get the committee's approval for the last valid state. 
\sys provides sub-second latency because it does not employ heavy-weight consensus. Instead, \sys uses \textit{consistent broadcast} to announce updates and close the channel, a light-weight abstraction that is powerful enough to preserve safety and liveness to any rational parties.
Furthermore, we consider permissioned blockchains, where the additional property of auditability might be desired for regulatory purposes. We introduce \asys, an off-chain construction that provides auditability on top of \sys without conflicting with its privacy guarantees. 
We formally define the properties our payment channel construction should fulfill, and prove that both \sys and \asys satisfy them. We also design incentives for \sys such that honest and rational behavior aligns. Finally, we provide a reference implementation of the smart contracts in Solidity.
\end{abstract}


\section{Introduction}
\label{sec:intro}

The prime solution to the scalability challenge~\cite{croman2016scaling} of large-scale blockchains, are the so-called \textit{channels}~\cite{spilman2013channels,DW2015channels,poon2015lightning}.
The idea is that any two parties that interact (often) with each other can set up a joint account on the blockchain, \ie, a channel. Using this channel, the two parties can transact off-chain, sending money back and forth by just sending each other signed messages. The two parties are relying on the blockchain as a fail-safe mechanism in case of disputes. \com{Channels can also be chained to allow for whole payment networks by applying so-called hash time lock contracts}

The security guarantees of a channel are ensured by a dispute handling mechanism. If one party tries to cheat the other party, in particular by trying to close a channel on the underlying blockchain in an invalid (outdated) state, then the attacked party has a window of time ($t_d$) to challenge the fraud attempt. 
Hence, a channel is secure as long as all parties of the channel are frequently -- at least once in $t_d$ time -- online and monitoring the blockchain.
This is problematic in real networks~\cite{miller2013feather}, as one party may simply execute a denial of service (DoS) attack on the other party. To add insult to injury, the parameter $t_d$ is public; the attacking party hence knows the exact duration of the denial of service attack.
\com{It has been shown that these attacks are feasible~\cite{miller2013feather}.} \com{With the proliferation of payment networks such as Lightning on top of Bitcoin system, these attacks have become relevant.}


The issue is well-known in the community, and there were solution attempts using semi-trusted third parties called \textit{watchtowers}~\cite{dryja2016monitors,watchtowers,avarikioti2018towards,mccorry2019pisa}.
The idea is that worrisome channel parties can hire watchtowers which watch the blockchain on their behalf in case they were being attacked. So instead of DoSing a single machine of the channel partner, the attacker might need to DoS the channel partner as well as its watchtower(s). This certainly needs more effort as  the adversary must detect a watchtower reacting and then block the dispute from appearing on-chain. However, if large amounts of money are in a channel, it will easily be worth the investment.

While DoS attacks are also possible in blockchains such as the Bitcoin blockchain, DoS attacks on channels have a substantially different threat level. A DoS attack on a blockchain is merely a liveness attack: One may prevent a transaction from entering the blockchain at the time of the attack. However, the parties involved with the transaction will notice this, and can simply re-issue their transaction later. A DoS attack on a channel, on the other hand, will steal all the funds that were in the channel. Once the fraudulent transaction is in the blockchain, uncontested for $t_d$ time, the attack succeeds, and nobody but the cheated party (and its watchtowers) will know any better.

Channels need a more fundamental solution. Not unlike blockchains, introducing timing parameters may be okay for liveness. Security on the other hand should be guaranteed even if the network behaves completely asynchronously. 




To address these issues we introduce \sys, a novel incentive-compatible payment channel construction that does not rely on any timing assumption for the delivery of messages to be secure.
\sys provides proactive security, detecting and preventing fraud before it appears on-chain. 
As a result, \sys can guarantee the security of channels even under censorship\footnote{This censoring ability is encompassed by the chain-quality property~\cite{garay15bitcoin} of blockchain systems which is rightly bound to the  synchrony of the network.}~\cite{miller2013feather} or any liveness\footnote{Liveness states that a transaction will be eventually included on-chain as long as it is provided to all honest parties for a long enough time period~\cite{garay15bitcoin}.} attack (see Appendix~\ref{app:attack}).

To achieve these properties, \sys needs to address three key challenges.
The first challenge is how to achieve this proactive check without using a single trusted third party that approves every transaction.
The core idea of \sys is to provide proactive security to the channel instead of reactive dispute resolution.
To this end, \sys employs a group of \guards.
If there is a dispute, the \guards  make sure the correct state is the only one available for submission on-chain, regardless of the amount of time it takes to make this final state visible. 

The second challenge for \sys is cost. To simulate this trusted third party, it would need the \guards to run costly asynchronous consensus~\cite{kokoris19bootstrapping} for every update. Instead, in \sys we show that a light-weight consistent broadcast protocol is enough to preserve both safety and liveness.

A final challenge of \sys that we address are incentives. While the \guards may be partially Byzantine, we additionally want honest behavior to be their dominant strategy.
Unfortunately, existing watchtower solutions do not align the expected and rational behavior of the watchtower, hence a watchtower is reduced to a trusted third party.
Specifically, Monitors~\cite{dryja2016monitors}, Watchtowers~\cite{watchtowers}, and DCWC~\cite{avarikioti2018towards} pay the watchtower upon fraud. 
Given that the use of a watchtower is public knowledge, any rational channel party will not commit fraud and hence the watchtower will never be paid. Therefore, there is no actual incentive for a third party to offer a watchtower service. 
On the other hand, Pisa \cite{mccorry2019pisa} pays the watchtower regularly every time a transaction is executed on the channel. The watchtower also locks collateral on the blockchain in case it misbehaves. However, Pisa's collateral is not linked to the channel or the party that employed the watchtower. Hence, a watchtower that is contracted by more than one channel can double-assign the collateral, making Pisa vulnerable to bribing attacks.
Even if the incentives of Pisa get fixed, punishing a misbehaving watchtower in Pisa is still a synchronous protocol (though for a longer period).
In \sys, we employ both rewards and punishment to design the appropriate incentives such that honest and rational behavior of \guards align, while no synchrony assumptions are required, \ie, the punishment of misbehaving \guards is not conditional on timing assumptions.

We also present \asys, a channel construction suitable for regulated environments such as supporting a fiat currency~\cite{wust2019prcash}.
\asys needs to additionally provide auditability of transactions without forfeiting privacy~\cite{gudgeon2019sok} and while preserving the accountability of the auditor.
To resolve this we change \sys in two ways. First, the state updates are no longer just private but also interconnected in a tamper-evident hash-chain.
Essentially, the \guards maintain a  single hash (constant-size storage cost) which is the head of the hash-chain of the state history.
Second, to provide accountability for the auditor, 
we require that the auditor posts the access request on-chain~\cite{kokoris18hidden}. Only then will the \guards provide the auditor the necessary metadata to verify the state history received from the parties of the channel.

To evaluate our channel construction we deploy our protocol on a large-scale testbed and show that
the overhead of an update is around the round-trip latency of the network (in our case $0.1$ seconds). 
Unlike existing channels, the parties in \sys need not wait for the dispute transaction to appear on-chain. Hence, our dispute resolution mechanism is three orders of magnitude faster than existing blockchain systems that need to wait until the transaction is finalized on-chain.
We additionally implement the on-chain operations of \sys in a Solidity smart contract that can be deployed on the Ethereum blockchain. We provide gas measurements for typical operations on the smart contract illustrating that it is practical. Our smart contract implementation is well tested and can be adopted towards a real deployment of \sys.
\com{The evaluation can be found in Section~\ref{sec:evaluation}.}

In summary, this paper makes the following contributions:
\begin{itemize}
    \item We introduce \sys, the first incentive-compatible off-chain construction that operates securely with offline channel participants under full asynchrony with sub-second latency. 
    \item We introduce \asys, a modified \sys channel construction, that enables external auditors to lawfully request access to the channels history while maintaining privacy.
    \item We define the desired channel properties and show they hold for \sys under a hybrid model of rational and byzantine participants (channel parties and \guards). Specifically, we present elaborate incentive mechanisms (rewards and punishments) for the \guards to maintain the channel properties under collusion or bribing.
    \item We evaluate the practicality of \sys by fully implementing its on-chain functionality in Solidity for the Ethereum blockchain. We measure its operational costs in terms of gas and illustrate that its deployment is practical.
\end{itemize}

\section{Background and Preliminaries}
In this section, we provide the necessary introduction to payment channels.
Furthermore, we introduce the necessary distributed abstractions 
that our constructions build upon.

\subsection{Blockchain Scalability \& Layer 2}
One of the major problems of blockchain protocols is the limited transaction throughput that is associated with the underlying consensus mechanism, originally introduced in Bitcoin~\cite{nakamoto2008bitcoin}. Nakamoto consensus demands that every participant of the system verifies and stores a replica of the entire history of transactions, \ie, the blockchain, to guarantee the safety and liveness of the transaction ledger. However, this requirement leads to blockchain systems with limited block size and block creation time. If we increase size or decrease time, we implicitly enforce participants to verify and store more data, which in turn leads to centralization and additional advantages for participants that invest in more infrastructure. Thus, blockchain systems that use the Nakamoto or similar consensus mechanisms face a scalability problem.   
Notably, Bitcoin handles at most seven transactions per second~\cite{croman2016scaling}, while current digital monetary systems, such as Visa, handle tens of thousands.

Proposed solutions for the throughput limitation of blockchain systems can be categorized in two groups. On-chain solutions that attempt to create faster blockchain protocols~\cite{kokoris17omniledger,kokoris16enhancing,decker16bitcoin}, and off-chain solutions that use the blockchain only as a fail-safe mechanism and move the transaction load offline, where the bottleneck is the network speed. While on-chain solutions lead to the design of new promising blockchain systems, they typically require stronger trust assumptions and they are not applicable to existing blockchain systems (without a hard fork). In contrast, off-chain (layer 2) solutions are built on top of the consensus layer of the blockchain and operate independently. Essentially, off-chain solutions allow two parties\footnote{Note that a channel can also be created between multiple parties~\cite{burchert2018scalable}.} to create a ``channel'' on the blockchain through which they can transact fast and secure; this solution is known as \emph{payment channel}~\cite{spilman2013channels,DW2015channels}.   

Payment channels allow transactions between two parties to be executed instantly off-chain while maintaining the guarantees of the blockchain. Essentially, the underlying blockchain acts as a ``judge'' in case of fraud.
There are multiple proposals on how to construct payment channels \cite{poon2015lightning,DW2015channels,deckereltoo,spilman2013channels}, but all proposals share the same core idea: The two parties create a joint account on the blockchain (funding transaction). Every time the parties want to make a transaction they update the distribution of the capital between them accordingly and they both sign the new  transaction as if it would be published on the blockchain (update transaction)\footnote{By definition, a channel requires the participation and signature of every party of the channel to update the channel's state in order to maintain security.}. To close the channel, a party publishes the latest update transaction either unilaterally or along with the counter-party with a closing transaction. 

The various proposals differ in the way they handle disputes, \ie, the case where one of the parties misbehaves and attempts to close the channel with a transaction that is not the latest update transaction, thus violating the safety property. 
Lightning channels~\cite{poon2015lightning} penalize the misbehaving party by assigning the money of the channel to the counter-party in case of fraud. To achieve this, every time an update transaction is signed each party releases a secret to the counterparty. That secret enables the counterparty to claim the money of the channel in case the party publishes the previous update transaction (breach remedy). However, this transaction is valid only for a window of time since the party should, in case of no fraud, eventually be able to spend its money from the channel. This dispute period is enforced with a (relative) timelock. 
On the other hand, Duplex channels~\cite{DW2015channels} guarantee that the latest update transaction will become valid before any previous update transaction, again utilizing timelocks.
In both cases, the liveness of the underlying blockchain and timelocks are crucial to the safety of the payment channel solution. Additionally, both solutions require online participants that frequently monitor the blockchain to ensure safety. 

\com{
To extend the concept of payment channels on an arbitrary state, \textit{state channels} were introduced~\cite{Miller2017sprites}. Several recent constructions exist in this direction~\cite{dziembowski2017perun, Miller2017sprites,coleman2018counterfactual}. In contrast with these works, we focus on designing a safe state channel protocol that is asynchronous, \ie, does not require timelocks, which can be used later as a \emph{brick} in any of these constructions.
}
\subsection{Consistent Broadcast}

Consistent broadcast~\cite{reiter94secure} is a distributed protocol run by a node that wants to send a message to a set of peers reliably.
It is called consistent because it guarantees that if a correct peer delivers
a message $m$ with sequence number $s$ and another correct peer delivers message $m'$ with sequence number $s$, then $m=m'$. 
Thus, the sender cannot equivocate.
In other words, the protocol maintains consistency among the delivered messages with the same sender and sequence numbers but makes no provisions that any parties do deliver the messages. 
In our system we only care about the consistency of sequence numbers, as any party of the channel can be the sender of a message $m$ even after $m$ is correctly broadcast.
We allow this in order to remove the need for parties to share the proof that the transaction is committed, as there is no incentive to do so.
\section{Protocol Overview}
In this section, we present the system and threat model we consider, a high-level overview of \sys and \asys together with the goals we want to achieve.

\subsection{System Model}
\label{sec:model}

We make the usual cryptographic assumptions: the participants are computationally bounded and cryptographically-secure communication channels, hash functions, signatures, and encryption schemes exist. 
We assume the underlying blockchain maintains a distributed ledger that has the properties of \textit{persistence} and \textit{liveness} as defined in \cite{garay15bitcoin}. 
We also assume that any message sent by an honest party will be delivered to any other honest party within a polynomial number of rounds. We do not make any additional assumptions (e.g., known bounds for message delivery).
Furthermore, we do not require a ``perfect'' blockchain system since \sys can tolerate temporary liveness attacks. Specifically, if an adversary temporarily violates the liveness property of the underlying blockchain, this may result in violating the liveness property of channels but will not affect the safety.

\subsection{Threat Model}
We initially assume that at least one party in the channel is honest to simplify the security analysis. However, later, we show that the security analysis holds as long as the ``richest'' party of the channel is rational and intentionally deviates from the protocol if it can increase its profit (utility function).
Regarding the committee, we assume that there are at most $f$ out of $n=3f+1$ Byzantine \guards and we define a threshold $t=2f+1$ to achieve the liveness and safety properties.
The non-byzantine part of the committee is assumed rational; we first prove the protocol goals for $t$ honest \guards, and subsequently align the rational behavior to this through incentives.

\subsection{\sys Overview}

Both parties of a channel agree on a committee of \guards before opening the channel. The \guards commit their identities on the blockchain during the funding transaction of the channel (opening of the channel). After opening the channel on the blockchain, the channel can only be closed either by a transaction published on the blockchain and signed by both parties or by a transaction signed by one of the parties and a threshold ($t$) of honest  \guards. Thus, the committee acts as power of attorney for the parties of the channel. 
Furthermore, \sys employs correct incentives for the $t$ rational \guards to follow the protocol, hence it can withstand $t=2f+1$ rational and $f$ Byzantine \guards,
while the richest channel party is assumed rational and the other byzantine.

A naive solution would then instruct the committee to run asynchronous consensus on every new update, which would cost $O(n^4)$~\cite{kokoris19bootstrapping} per transaction, a rather big overhead for the critical path of our protocol\footnote{For this step, we would require a fully asynchronous consensus protocol without a trusted setup. To the best of our knowledge, \cite{kokoris19bootstrapping} has the lowest communication complexity of $O(n^4)$.}.
Instead in Brick, consensus is not necessary for update transactions, as we only provide guarantees to rational parties (if both parties misbehave one of them might lose its funds).
As a result, every time a new update state occurs in the channel (i.e., a transaction), the parties run a consistent broadcast protocol (cost of $O(n)$) with the committee. Specifically, a party announces to each \guard that a state update has occurred. This announcement is a monotonically increasing sequence number to guarantee that the new state is the freshest state, signed by both parties of the channel to signal that they are in agreement. 
If the consistent broadcast protocol succeeds ($t$ \guards acknowledge reception) then this can serve as proof for both parties that the state update is safe.
After this procedure terminates correctly, both parties proceed to the execution of the off-chain state.

At the end of the life-cycle of a channel, a dispute might occur, leading to the unilateral closing of the channel. Even in this case, we can still guarantee the security and liveness of the closing with consistent broadcast.
The crux of the idea is that if $2f+1$ \guards accepted the last sequence number before receiving the closing request (hence the counterparty has committed), then at least one honest \guard will not accept the closing at the old sequence number. Instead, the \guard will reply to the party that it can only close at the state represented by the last sequence number.
As a result we define a successful closing to be at the maximum of all proposed states, which guarantees safety.
\textit{Although counter-intuitive, this closing process is safe because the transactions are already totally ordered and agreed to by the parties of the channel;} thus, the committee simply acts as shared memory of the last sequence number.

\subsection{\asys Overview}

\asys is designed to enable payment channels in a permissioned, regulated setting, for example, a centrally-banked cryptocurrency. In such a setting, there will be an auditor (\eg,~the IRS) that can check all the transactions inside a channel as these transactions might be taxable or illegal. This is a realistic case as the scalability in payment channels comes from a persistent relationship that models well B2B and B2C relationships that are usually taxed. 
In this setting, we assume that the auditor can punish the parties and the committee externally of the system, hence our goal is to enhance transparency even if misbehavior is detected retroactively.

In order to convert \sys into \asys we must ensure that (a) nothing happens without the committee's approval and (b) a sufficient audit trail is left on-chain to stop regulators from misbehaving.
We resolve the first issue by disabling the parties' ability to close the channel without the participation of the committee. For the second challenge, the parties generate a hashchain of (blinded) states and sequence numbers and \guards always remember the head of the hashchain together with the last sequence number.
To prevent the auditor from misbehaving, we force the auditor to post a ``lawful access request''~\cite{feigenbaum17multiple} on-chain to convince the channel parties to initiate the closing of the channel for audit purposes and send the state history to the auditor.

\subsection{
Reward Allocation \& Collateral}
To avoid bribing attacks, we enforce the \guards to lock  collateral in the channel. The total amount of collateral is proportional to the value of the channel meaning that if the committee size is large, then the collateral per \guard is small. 
More details on the necessary amount of collateral are thoroughly discussed in Sections \ref{sec:incentives} and \ref{sec:limit}.
Additionally, the committee is incentivized to actively participate in the channel with a small reward that each \guard gets when they acknowledge a state update of the channel. 
This reward is given with a unidirectional channel~\cite{gudgeon2019sok}, which does not suffer from the problems \sys solves.
Moreover, the \guards that participate in the closing state of the channel get an additional reward, hence the \guards are incentivized to assist a party when closing in collaboration with the committee is necessary.

\subsection{Protocol Goals}
\label{subsec:properties}
To define the goals of \sys, we first need to define the necessary properties of a channel construction. 
Intuitively, a channel should ensure similar properties with a blockchain system, \ie, a party cannot cheat another party out of its funds, and any party has the prerogative to eventually spend its funds at any point in time. 
The first property, when applied to channels, means that no party can cheat the channel funds of the counterparty, and is encapsulated by \textit{Safety}. 
The second property for a channel solution is captured by \textit{Liveness}; it translates to any party having the prerogative to eventually close the channel at any point in time.
We say that a channel is \textit{closed} when the locked funds of the channel are spent on-chain, while a channel \textit{update} refers to the off-chain change of the channel's state.

In addition, we define \textit{Privacy} which is not guaranteed in many popular blockchains, such as Bitcoin \cite{nakamoto2008bitcoin} or Ethereum \cite{wood2014ethereum}, but constitutes an important practical concern for any functional monetary (cryptocurrency) system.
Furthermore, we define another optional property, \textit{Auditability}, which allows authorized external parties to audit the states of the channel, thus making the channel construction suitable for use on a regulated currency. 
The first three properties are met by \sys, while the latter is only available in \asys.

First, we define some characterizations on the state of the channel, namely, validity and commitment. Later, we define the properties for the channel construction.
Each state of the channel has a discrete sequence number that reflects the order of the state. We assume the initial state of the channel has sequence number $1$ and every new state has a sequence number one higher than the previous state agreed by both parties. We denote by $s_i$ the state with sequence number $i$.

\begin{definition}
A state of the channel, $s_i$, is \textbf{valid} if the following hold:
\begin{itemize}
    \item Both parties of the channel have signed the state $s_i$.
    \item The state $s_i$ is the \textbf{freshest} state, \ie, 
    no subsequent state $s_{i+1}$ is valid.
    \item The committee has not invalidated the state. The committee can invalidate the state $s_i$ if the channel closes in the state $s_{i-1}$.
\end{itemize}
\end{definition}

\begin{definition}
A state of the channel is  \textbf{committed} if it was signed by at least $2f+1$ \guards or is valid and part of a block in the persistent\footnote{The part of the chain where the probability of fork is negligible hence there is transaction finality, \eg, 6 blocks in Bitcoin.} part of the blockchain.
\end{definition}

\begin{definition}[\textbf{Safety}]
A \sys channel will only close in the freshest committed state.
\end{definition}
    
\begin{definition}[\textbf{Liveness}]
Any valid operation (update, close) on the state of the channel  will eventually\footnote{Depending on the message delivery.} be committed (or invalidated).
\end{definition}    
 
\begin{definition}[\textbf{Privacy}]
 No unauthorized~\footnote{Authorized parties are potential auditors of the channel, as described in \asys.} external (to the channel) party learns about the state of the channel (\eg,~the current distribution of funds between the parties of a payment channel) unless at least one of the parties initiate the closing of the channel.
\end{definition} 

\begin{definition}[\textbf{Auditability}]
 All committed states of the channel are verifiable by an authorized third party.
\end{definition}

\section{\sys Design}
In this section, we first present the \sys architecture, and later the additional incentive mechanisms.

\begin{figure*}[t]
    \vspace{-0.3cm}
    \centering
    \includegraphics[width=0.55\textwidth]{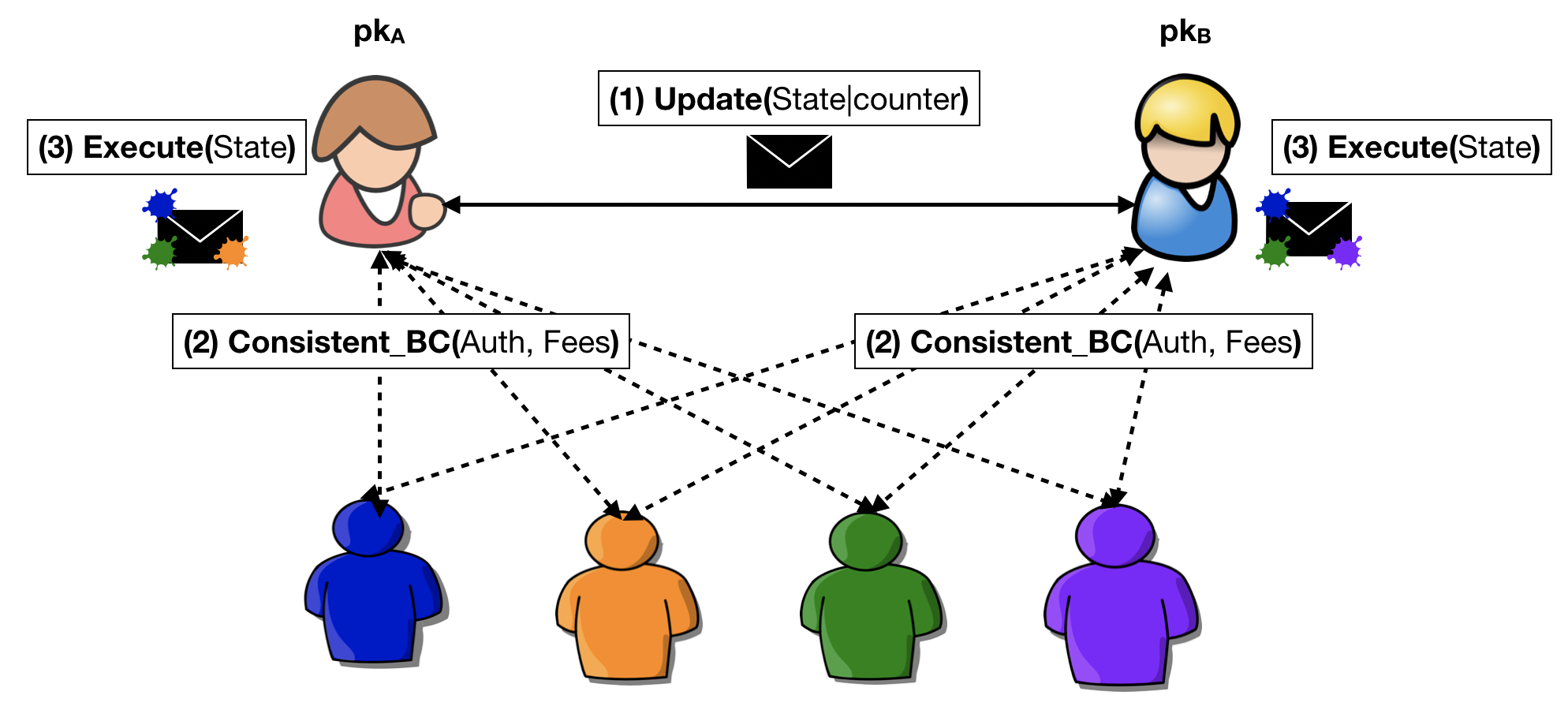}
    \vspace{-0.2cm}
    \caption{ \scriptsize Typical workflow of \sys for a state update. (1) Alice and Bob agree on a new state update. (2) They individually broadcast that they reached agreement on a new state to the committee (with an associated fee). (3) When a threshold of the committee replies that the state is committed,  each sender (Alice or Bob) executes this update as persistent.}
    \label{fig:robustness}
    \vspace{-0.7cm}
\end{figure*}

\sys consists of three phases: \textit{Open}, \textit{Update}, and \textit{Close}. 
We assume the existence of a \textit{smart contract} (self-executed code run on a blockchain). The \sys smart contract has two functions, \textit{Open} and \textit{Close}, that receive the inputs of the protocols and verify that they adhere to the abstractly defined protocols specified below.

Algorithm \ref{alg:open} describes the first phase,
\textit{Open}, which is the opening of a channel between two parties. In this phase, the parties create the initial funding transaction, similarly to other known payment channels such as \cite{DW2015channels,poon2015lightning}. 
However, in \sys we also define three additional parameters during this phase; we include in the funding transaction the hashes of the public keys of the \guards of the channel, denoted by $W_1, W_2,\dots, W_n$, the threshold $t$, and a closing fee $F$.
This fee will be awarded to the responsive \guards in the last phase,
\textit{Close},
if and only if the closing of the channel is done in collaboration with the committee. In this case, the first $t$ \guards that assisted in closing the channel will be rewarded with the amount $F/t$. In addition, each \guard locks a collateral in the \sys smart contract which will be claimed by the parties of the channel during phase \textit{Close} if a \guard misbehaves (\eg,~receives a bribe). 

\begin{protocol}
\DontPrintSemicolon
\SetAlgoLined
\textit{Input:} Channel parties $A,B$, \guards $W_1, W_2, \dots, W_n$, initial state $s_1$, closing fee $F$. \;
\textit{Goal:} Open a \sys payment channel. \; 
\vspace{0.25cm}
1. Both parties  $A,B$ sign:
$open(H(W_1), H(W_2), \dots, H(W_n), t, s_1, F)$. \;
2. Register to $\{M, \sigma(M)\}$ the announcement of Protocol \texttt{Update}($A, B, s_1$).\;
3. Execute Protocol \texttt{Consistent Broadcast}($M, \sigma(M), A, B, W_1, W_2, \dots, W_n$). \;
4. Publish to the blockchain $\sigma(open(H(W_1), H(W_2), \dots, H(W_n), t, F))$.\;
\caption{\texttt{Open}}
\label{alg:open}
\end{protocol}


The second phase of the protocol consists of two sub-protocols, \texttt{Update} (Protocol \ref{alg:update}) and \texttt{Consistent Broadcast} (Protocol \ref{alg:broadcast}). Both algorithms are executed consecutively every time an update occurs, \ie, when the state of the channel changes. During \texttt{Update}, the parties of the channel agree on a new state and create an announcement, which they to the committee using \texttt{Consistent Broadcast}. 
To agree on a new state, both parties sign the hash of the new state\footnote{Blinding the commitment to the state is not necessary for \sys, but we do it for compatibility with \asys.}. 
This way both  parties commit to the new state of the channel, while none of the parties can unilaterally close the channel without the collaboration of either the counterparty or the committee.
The announcement, on the other hand, is the new sequence number signed by both parties of the channel\footnote{We abuse the notation of signature $\sigma$ to refer to the multisig of both $A$ and $B$.}.
\emph{The signed sequence number allows the \guards to verify agreement has been reached between the channel parties on the new state, while the state of the channel remains private.}
Furthermore, if a party wants to close the channel unilaterally, the \guards can provide a signature on the stored announcement $\sigma_{W_i}(M,close)$. 
Consequently, the \sys smart contract, given the signed hash of the freshest state by a party, can verify that the state corresponds to the maximum sequence number submitted by the \guards and the presence of the appropriate signatures, and then close the channel in the correct state.

In Protocol~\ref{alg:broadcast}, for every state update, each party sends to all \guards the announcement including a small \textit{update fee} for watching the channel. Then, each \guard replies to every party that sent the announcement by signing the announcement.
The \guard's signature can be used later to penalize the \guard (a party can claim the \guard's collateral) in case the \guard colludes with a party and signs a previous state to close the channel.

\begin{protocol}
\DontPrintSemicolon
\SetAlgoLined
\textit{Input:} Channel parties $A,B$, current state $s$. \;
\textit{Goal:} Create announcement $M,\sigma(M)$ (sequence number of new state signed by both parties). \; 
\vspace{0.25cm}
1. Both parties $A,B$ sign and exchange: $\{H(s_i,r_i), i\}$, where $r_i$ is a random number and $s_i$ the current state.\;
2. After receiving the signature of the counterparty on $\{H(s_i,r_i), i\}$, each party sends to the counterparty its signature on the sequence number $\sigma(i)$ thus creating the announcement $\{M, \sigma(M)\} (M=i)$.
\caption{\texttt{Update}}
\label{alg:update}
\end{protocol}

\begin{protocol}
\DontPrintSemicolon
\SetAlgoLined
\textit{Input:} Channel parties $A,B$, \guards $W_1, W_2, \dots, W_n$, announcement $\{M, \sigma(M)\}$. \;
\textit{Goal:} Inform the committee of the new update state and verify the validity of the new state. \; 
\vspace{0.25cm}
1. Each party broadcasts to all the \guards $W_1, W_2, \dots, W_n$ the announcement $\{M, \sigma(M)\}$ alongside a fee $r$.\;
2. Each \guard $W_j$, upon receiving $\{M, \sigma(M)\}$, verifies that both parties' signatures are present, and the sequence number is exactly one higher than the previously stored sequence number. 
If the \guard has published a closing state, it ignores the  state update. 
Otherwise, $W_j$ stores the announcement $\{M, \sigma(M)\}$ (replacing the previous announcement), signs $M$, and sends the signature $\sigma_{W_j}(M)$ to the parties that payed the update fee $r$.\;
3. Each party, upon receiving at least $t$ signatures on the announcement $M$, considers the state committed and proceeds to the state transition.\;
\caption{\texttt{Consistent Broadcast}}
\label{alg:broadcast}
\end{protocol}

The last phase of the protocol, \textit{Close}, can be implemented in two different ways: the first is similar to the traditional approach for closing a channel (Protocol \ref{alg:optclose}: \texttt{Optimistic Close}) where both parties collectively sign the freshest state (closing transaction) and publish it on-chain. However, in case a channel party is not responding to new state updates or closing requests, the counterparty can unilaterally close the channel in collaboration with the committee of the channel. 
In Protocol \ref{alg:pesclose}: \texttt{Pessimistic Close}, a party initiates the last phase of the protocol by requesting from each \guard its signature on the last committed sequence number.
A \guard, upon receiving the closing request, publishes on-chain a closing announcement, \ie, the stored sequence number signed along with a flag close.
When $t$ closing announcements are on the persistent part of the chain, the party recovers the state that corresponds to the maximum sequence number from the closing announcements $s_{i}$. Then, the party publishes state $s_{i}$ and the random number $r_{i}$ along with the signatures of both parties on the corresponding hash and sequence number $\sigma(H(s_{i},r_{i}),i)$ on-chain. 
As soon as these data are included in a (permanent), the \sys smart contract verifies and closes the channel. In particular, the smart contract (a)~recovers from the submitted state and salt the hash $H(s_{i},r_{i})$ and the maximum sequence number $i$, (b)~verifies that the signatures of both parties are on the message $\{H(s_{i},r_{i}),i\}$, and (c)~there are $t$ submitted announcements that correspond to \guard identities committed on-chain in Protocol~\ref{alg:open}. If all verifications check the smart contract closes the channel in the submitted state and rewards $F/t$ to the first $t$ participating \guards.

\begin{protocol}
\DontPrintSemicolon
\SetAlgoLined
\textit{Input:} Channel parties $A,B$, state $s$. \;
\textit{Goal:} Close a channel on state $s$, assuming both parties are responsive and in agreement. \; 
\vspace{0.25cm}
1. A party $p \in \{A,B\}$ broadcasts the request $close(s)$.\;
2. Both parties $A,B$ sign the state $s$ (if they agree) and exchange their signatures. \;
3. The party $p$ (or any other channel party) publishes the signed by both parties state, $\sigma_{A, B}(s)$ on-chain.\;
\caption{\texttt{Optimistic Close}}
\label{alg:optclose}
\end{protocol}

\begin{protocol}
\DontPrintSemicolon
\SetAlgoLined
\textit{Input:} Party $p \in \{A, B\}$, \guards $W_1, W_2, \dots, W_n$, state $s_i$, random number $r_i$.
\;
\textit{Goal}: Close a channel on state $s_i$ with the assist of the committee. \; 
\vspace{0.25cm}
1. Party $p$ broadcasts to the \guards $W_1, W_2, \dots, W_n$ the request $close()$. \;
2. Each \guard $W_j$ publishes on-chain a signature on the (last) stored announcement $\sigma_{W_j}(M, close)$ and stops signing new state updates.\;
3. Party $p$, upon verifying $t$ on-chain signed announcements by the \guards, recovers the $max(i)$ that is included in the announcements. Then, party $p$  publishes on-chain the state $s_i$, the random number $r-i$ and the signature of both parties on $\{H(s_{i},r_{i}),i\}$ \;
4. After the state is included in a (permanent) block, the smart contract recovers $\{H(s_{i},r_{i}),i\}$, verifies both parties' signatures and the \guards identities, and then closes the channel in state $s_i$.
The first $t$ \guards whose signature are published on-chain get an equal fraction of the closing fee $F/t$. \;
\caption{\texttt{Pessimistic Close}}
\label{alg:pesclose}
\end{protocol}

\subsection*{Incentivizing Honest Behavior}\label{sec:incentives}

\sys actually works without the fees, if we assume one honest party and $t$ honest \guards. However, our goal is to have no honest assumptions and instead align rational behavior to honest through incentives. There are three incentive mechanisms in \sys: 
\begin{enumerate}
    \item \textit{Update Fee ($r$). }The parties establish a one-way channel~\cite{gudgeon2019sok} with each \guard where they send a small payment every time they want a signature for a state update. Note that the update fee is awarded to the \guards at step 1 of Protocol~\ref{alg:broadcast}.
    \item \textit{Closing Fee ($F$). }
    During phase \texttt{Open} (Protocol \ref{alg:open}), the parties lock a closing fee $F$ in the channel, which is split only among the  first $t$ \guards that participate in Protocol \ref{alg:pesclose}. 
    If the channel closes optimistically (Protocol~\ref{alg:optclose}), the closing fee returns to the parties.
\item \textit{\guard Collateral. }During phase \textit{Open}, each \guard locks collateral at least equal to the amount locked in the channel $v$ divided by $f$. The collateral is either returned to the \guard at the closing of the channel or claimed by a channel party that provides a proof-of-fraud. 
    A proof-of-fraud consists of two conflicting messages signed by the same \guard:  (a)~a signature on an announcement on a state update of the channel, and (b)~a signature on an announcement for closing on a previous state of the channel.
    
    In case, a party submits the closing announcements and at most $f$ proofs-of-fraud, to close the channel we execute a second closing process excluding the \guards that committed fraud. Then, the channel closes in the state with the maximum sequence number of the announcements submitted by a total of $t$ non-excluded \guards. On the other hand, if a party submits the closing announcements and at least $f+1$ proofs-of-fraud, the party that submitted the proofs-of-fraud claims only the collateral and the entire channel balance is awarded to the counterparty.
    Lastly, if no proofs-of-fraud are submitted the channel closes as described in Protocol~\ref{alg:pesclose}.
 \end{enumerate}

We further demand that the size of the committee is at least $n > 7$, hence $f>2$. As a result, we guarantee there is at least one channel party with locked funds greater than each individual \guard's collateral, $\frac{v}{2} > \frac{v}{f}$.
This restriction along with the aforementioned incentive mechanisms ensure resistance to collusion and bribing of the committee, meaning that following the protocol is the dominant strategy for the rational \guards.

\section{\sys Analysis}\label{sec:security}
We first prove \sys satisfies \textit{Safety} and \textit{Liveness} assuming at least one honest channel party and at least $t$ honest \guards. 
Furthermore, we note that \sys achieves \textit{Privacy} even if all \guards are byzantine while the channel parties are rational.
Then, we show that rational \guards that want to maximize their profit will follow the protocol specification, for the incentive mechanisms presented in Section~\ref{sec:incentives}.
Essentially, we show that \sys enriched with the proposed incentive mechanisms is dominant-strategy incentive-compatible.
We complete our analysis by showing that \sys security holds as long as the richest channel party is rational in Appendix~\ref{sec:rationalparties}.

\subsection{Security Analysis under one Honest Participant and $t$ Honest \guards}
 \begin{theorem}\label{thm:safety}
 \emph{\sys} achieves safety under asynchrony assuming one byzantine party and $f$ byzantine \guards.
 \end{theorem}
 
 \begin{proof}
    In \sys there are two ways to close a channel (phase \textit{Close}), either by invoking Protocol~\ref{alg:optclose} or by invoking Protocol~\ref{alg:pesclose}. 
    In the first case (\texttt{Optimistic Close}), both parties agree on closing the channel in a specific state (which is always the freshest valid state \footnote{We assume that if the parties want to close the channel in a previous state, they will still create a new state similar to the previous one but with an updated sequence number.}). As long as this valid state is published in a block in the persistent part of the blockchain, it is considered to be committed. Thus, safety is guaranteed.
    
    In the second case, when Protocol~\ref{alg:pesclose}: \texttt{Pessimistic Close} is invoked, a party has decided to close the channel unilaterally in collaboration with the committee. 
    The \sys smart contract verifies that the state is valid, \ie, the signatures of both parties are present and the sequence number of that state is the maximum from the submitted announcements. Given the validity of the closing state, it is enough to show that the channel cannot close in a state with a sequence number smaller than the one in the freshest committed state. This holds because even if the channel closes in a valid but not yet committed state with sequence number larger than the freshest committed state, this state will eventually become the freshest committed state (similarly to Protocol~\ref{alg:optclose}).
    
    Let us denote by $s_i$ the closing state of the channel. Suppose that there is a committed state $s_k$ such that $k>i$, thus $s_i$ is not the freshest state agreed by both parties. We will prove safety by contradiction. 
    For the channel to close at $s_i$ at least $t=2f+1$ \guards have provided a signed closing announcement, and the maximum sequence number of these announcements is $i$. Otherwise the \sys smart contract would not have accepted the closing state as valid.
    According to the threat model, at most $n-t=f$ \guards are Byzantine, thus at least $f+1$ honest \guards have submitted a  closing announcement. Hence, none of the $f+1$ honest \guards have received and signed the announcement of any state with sequence number greater than $i$. 
    However, in phase~\textit{Update \& Consistent Broadcast}, an update state is considered to be committed, according to Protocol~\ref{alg:broadcast} (line 3), if and only if it has been signed by at least $t=2f+1$ \guards. Since at most $n-(f+1)=3f+1-f-1=2f < 2f+1$ \guards have seen (and hence signed) the state $s_k$, the state $s_k$ is not committed. Contradiction.
    
    We note that safety is guaranteed even if both parties crash. This holds because a state update requires unanimous agreement between the parties of the channel, \ie, both parties sign the hash of the new state.
\end{proof}

 \begin{restatable}{theorem}{liveness}\label{thm:liveness}
 \emph{\sys} achieves liveness under asynchrony assuming one byzantine party and $f$ byzantine \guards.
 \end{restatable}
\begin{proof}
We will show that every possible valid operation is either committed or invalidated. There are two distinct operations: \textit{close} and \textit{update}. 
We say that operation close applies in a state $s$ if this state was published on-chain either in collaboration of both parties (Protocol~\ref{alg:optclose}) or unilaterally by a channel party as the closing state (Protocol~\ref{alg:pesclose}, step 3).

If the operation is close and not committed, either the parties did not agree on this operation (\texttt{Optimistic Close}), or a verification of the smart contract failed (\texttt{Pessimistic Close}). In both cases, the operation is not valid.

Suppose now the operation is close and never invalidated. Then, if it is an optimistic close, all the parties of the channel have signed the closing state since it is valid. Since at least one party is honest the transaction will be broadcast to the blockchain.
Assuming the blockchain has liveness, the state will be eventually included in a block in the persistent part of the blockchain and thus will be eventually committed.
On the other hand, if it is a pessimistic close and not invalidated, the smart contact verifications were successful therefore the state was committed.

Suppose the operation is a valid update and it was never committed. 
Since the operation is valid and at least one party of the channel is honest, the \guards eventually received the state update (\texttt{Consistent Broadcast}). However, the new state was never committed, therefore at least $f+1$ \guards did not sign the update state. We assumed at most $f$ Byzantine \guards, hence at least one honest \guard did not sign the valid update state. According to Protocol~\ref{alg:broadcast} (line 2), an honest \guard does specific verifications and if the verifications hold the \guard signs the new state. Thus, for the honest \guard that did not sign, one of the verifications failed. If the first verification fails, then a signature from the parties of the channel is missing thus the state is not valid. Contradiction. The second verification concerns the sequence number and cannot fail, assuming at least one honest channel party. 
Thus, the \guard has published previously a closing announcement on-chain and ignores the state update. 
In this case, either (a)~the closing state of the channel is the new state - submitted by another \guard that received the update before the closing request - or (b)~the closing state had a smaller sequence number from the new state. In the first case (a), the new state is committed eventually (on-chain), while in the second case (b) the new state is invalidated as the channel closed in a previous state.

For the last case, suppose the operation is a valid update and it was never invalidated. We will show the state update was eventually committed. Suppose the negation of the argument towards contradiction. We want to prove that an update state that is not committed is either not valid or invalidated. The reasoning of the proof is similar to the previous case.
\end{proof}

Lastly, \sys achieves privacy even against byzantine \guards since they only receive the sequence number of each state update in a channel. Therefore, as long as parties do not intentionally reveal any information to anyone external, privacy is maintained.
We note, however, that in a network with multiple channels each channel needs to maintain a unique id which will be included in the announcement to avoid replay attacks.

\subsection{Incentivizing Rational \guards}
\label{sec:rationalwatchtowers}
In this section, we show that rational watchtowers that want to maximize their profit follow the protocol, \ie, that deviating from the honest protocol execution can only result in decreasing a \guard's expected payoff. 
We consider each phase of \sys separately, and evaluate the \guard's payoff for each possible action. 

\subsubsection*{Open} 
During the opening of the channel, we assume the \guards follow Protocol~\ref{alg:open} and commit the requested collateral on-chain (\sys smart contract). Otherwise, the parties simply replace the unresponsive/misbehaving \guards.
 
\subsubsection*{Update}
In Protocol~\ref{alg:update} the \guards do not participate. 
Thereby, the only action by the \guards is step 2 of Protocol~\ref{alg:broadcast}. A \guard can deviate as follows:
\begin{itemize}
    \item The \guard can simply ignore the party's request to  update the state (does not reply to the party), since the update fee is already collected.
    At first sight, this game looks like a fair exchange game, which is impossible to solve without a trusted third party~\cite{dziembowski18fairswap}.
Furthermore, we cannot use a blockchain to solve it~\cite{pagnia99impossibility} as the whole point of state channels is to reduce the number of transactions that go on-chain.
Fortunately, the state update game is a repeated game where \guards want to increase their expected rewards in the long term. Thereby, they know that if they receive an update fee from a party and do not respond, then the party will stop using them (there is $f$ fault tolerance in \sys), thus their expected payoff for the repeated game will decrease.
\item The \guard sends its signature on the new state but does not store the new announcement. 
In this case, if a party requests to close unilaterally the \guard cannot participate and thence collect the closing fee. Therefore, the \guard's expected payoff decreases.
\item The \guard replies to the party although it has published a closing announcement on-chain. Then, the party can penalize the \guard by claiming its collateral.
\item The \guard does not perform the necessary verifications. Then, the \guard might unintentionally commit fraud by signing an invalid state, hence allow the party to claim the collateral.
\end{itemize}
  
\subsubsection*{Close}
\guards do not participate in the optimistic close (Protocol~\ref{alg:optclose}).
As a result, we only evaluate the possible actions a \guard can deviate from the honest execution of Protocol~\ref{alg:pesclose}. 
\begin{itemize}
    \item The \guard does not publish a closing announcement on-chain. 
    In this case, the \guard can attempt to enforce a hostage situation on the funds of the channel in collaboration with  other \guards in order to blackmail the channel parties.
    However, to enforce a hostage situation on the channel's funds, at least $f+1$ \guards must collude, hence at least one rational \guard must participate. However, a rational \guard cannot be certain that the other \guards will indeed maintain the hostage situation or participate in the consensus thus claim the closing fee (only the first $t$ \guards get paid); therefore, we reduce our problem to the prisoner's dilemma problem. As a result, the only strong Nash equilibrium for a rational \guard is to immediately publish a closing announcement on-chain in order to claim later the closing fee.
    \item The \guard signs later a new state update. Then, the \guard allows the party receiving the signed announcement with higher sequence number than the closing announcement to create a proof-of-fraud and claim the \guard's collateral.
    \item The \guard publishes on-chain a closing announcement that is not the stored one\footnote{We assume the \guard will only sign as closing an announcement that used to be valid in an attempt to commit fraud. Otherwise, the party will claim the \guard's when the smart contract verifications fail.}. However, the \guard has already sent a signature on an announcement with a higher sequence number, therefore the party can create a proof-of-fraud and claim the \guard's collateral. 
    Hence, any rational \guard will request as a bribe an amount at least marginally higher than the collateral to perform the fraud. 
    We will show that no rational party will provide such a bribe to any rational \guard. Thus, all rational \guards will honestly follow the protocol and submit the stored announcement for closing the channel. 
    
    To that end, let us denote by $p_A$ the payoff function of party $A$ (the cheating party wlog). The payoff function depends on the channel balance of party $A$ when requesting to close the channel, denoted by $c_A$ ($0\geq c_A \geq v$, where $v$ is the total channel funds), the collateral the party claims through proofs-of-fraud, and the total amount spent for bribing rational \guards.
    Formally, \[p_A = c_A + x\frac{v}{f} - b\big(\frac{v}{f}+\epsilon\big)\]
    where $x$ is the number of proofs-of-fraud submitted by the party, $b$ the number of bribed rational \guards, and $\epsilon$ the marginal gain over the collateral the \guards require to be bribed. Note that each \guard has locked $\frac{v}{f}$ as collateral. Further, note that the Byzantine \guards do not require a bribe but act arbitrarily malicious, meaning that they will provide a proof-of-fraud to party $A$ without any compensation.
    
    Next, we analyze all potential strategies for party $A$ and demonstrate that the payoff function maximizes when the party does not bribe any rational \guard, but closes the channel in the freshest committed state.
    There are four different outcomes in the strategy space of party $A$:
    \begin{enumerate}
        \item The channel closes using Protocol~\ref{alg:pesclose}. Then, $p_A = c_A $.
        \item Party $A$ submits the proofs-of-fraud only from the Byzantine \guards and the channel closes in the freshest committed state (by the remaining $t$ rational \guard). Then, $p_A = c_A + f \frac{v}{f} = c_A + v$. Note that the payoff function in this case maximizes for $x=f$.
        \item Party $A$ bribes $b > 0$ rational \guards and the channel closes in the freshest committed state. Then, $p_A \leq c_A + (f+b)\frac{v}{f} - b(\frac{v}{f}+\epsilon) = c_A + v -b\epsilon \leq c_A + v - \epsilon.$
        The first inequality holds because the bribed \guards might be part of the second set of $t$ closing announcements, in which case they do not contribute to the claimed collateral.
        \item Party $A$ bribes $b > 0$ rational \guards and the channel closes in a state other than the freshest committed state. 
        Let us denote by $y$ the number of rational \guards that provide a proof-of-fraud to the party, and $m$ the number of submitted proofs-of-fraud that belong to Byzantine \guards. Then, the possible actions in this strategy are depicted in Table~\ref{tab:actions}. 
        \begin{table}[ht]
        \centering
        \caption{Potential actions to close the channel in a ``fraudulent'' state.}
        \begin{tabular}{l|c|c|c}
            \toprule
                    \textbf{Action}  & \textbf{Proof-of-fraud} & \textbf{Close} & \textbf{Total}\\
            \midrule
                   \textbf{Byzantine} & $m$  & $f-m$ & $f$  \\
            \midrule
                    \textbf{Bribed} & $y$ & $f+1-(f-m)$ & $y+m+1$\\
                    \textbf{(rational)}         & &   $=m+1$ &    \\
            \midrule
                    \textbf{Total} & $m+y$ & $f+1$ & - \\
            \bottomrule
        \end{tabular}
        \label{tab:actions}
        \end{table}
        Note that at least $f+1$ misbehaving \guards are required to close the state in a previous state since we assume there can be $f$ slow rational \guards that have not yet received the update states.
        In this case, the payoff function is $$p_A \leq v + \big(m+y\big)\frac{v}{f} - \big(y+m+1\big)\big(\frac{v}{f}+\epsilon\big)$$
        $$= v - \frac{v}{f} - \epsilon\big(y+m+1\big) \leq v-\frac{v}{f}-\epsilon$$
    \end{enumerate}
    where the first inequality holds since $0 \geq c_A \geq v$.
    Therefore, the payoff function maximizes in case the party follows the second strategy, \ie, when the channel closes in the freshest committed state and no rational \guards are bribed.
\end{itemize}
We notice that in any possible deviation of the protocols the \guard's expected payoff decreases, thus the  dominant strategy of a rational \guard is to follow the protocols.

\section{\asys Design}
In this section, we describe the \asys system design, while the security analysis for \asys can be found in Appendix~\ref{app:plus}.


\asys consists, similarly to \sys, of three phases, \textit{Open}, \textit{Update}, and \textit{Close}. However, \asys has one additional functionality, \textit{Audit}, that allows an authorized third party to audit the states of a channel. Invoking this functionality though enforces the closing of the channel.
The \textit{Audit} functionality is illustrated in Figure~\ref{fig:audit} (Appendix~\ref{app:plus}) and described in detail in Protocol~\ref{alg:audit}. 

To enable the \textit{Audit} functionality and therefore achieve \textit{Auditability}, we disable Protocol~\ref{alg:optclose}: \texttt{Optimistic Close}, and enforce the parties to close in collaboration with the committee. This way we guarantee that all states of the channel are available to the committee and hence to the potential auditor for verification.
Moreover, we modify Protocol~\ref{alg:update} and Protocol~\ref{alg:broadcast} (phase \textit{Update}) such that the \guards store a hash-chain of the state history instead of the sequence number of the freshest valid state they received. Thereby, we ensure that the parties cannot present an alternate state history to the auditor as it achieves fork-consistency~\cite{mazieres2002building}.

 The \textit{Audit} functionality is initiated by an authorized third party, the auditor, who publishes an access request on-chain. 
 Then, the \guards of the channel verify the validity  of the access request and initiate the closing of the channel by publishing the closing announcements (head of the hash-chain of the state history) on-chain (Protocol~\ref{alg:pesclose}: \texttt{Pessimistic Close}). 
 After the execution of Protocol~\ref{alg:pesclose}, the closing state is on-chain and both the (honest) \guards and parties of the channel have a consistent view of the channel history. 
Both parties send to the auditor the entire state history.
The auditor then verifies the state history received from both parties by computing the hash-chain and comparing the last hash with the hash that corresponds to the maximum sequence number published by the \guards (\ie, the hash that corresponds to the closing state of the channel). If the parties misbehave and send an alternate state history the auditor can pursue external punishment (\eg,~legal action).

\begin{protocol}
\DontPrintSemicolon
\SetAlgoLined
\textit{Input:} Auditor $A$ of channel $c$, audit smart contract with access on the information published on the blockchain, from Protocol~\ref{alg:open}.
\;
\textit{Goal}: Audit of the channel. \; 
\vspace{0.25cm}
1. The auditor $A$ publishes on-chain the access request for channel $c$.\;
2. Each \guard of the channel, upon verifying the validity of the access request, initiates the closing phase of the channel (Protocol~\ref{alg:pesclose}), and publishes on chain the stored announcement (head of the hash-chain of state history).\;
3. After the execution of Protocol \ref{alg:pesclose}, the channel is closed on-chain in a committed state $s$. Then, both parties of the channel send to the auditor the state history.\;
4. The auditor collects from the blockchain the hashes published by the \guards and selects the hash that corresponds to the maximum sequence number (\ie, corresponds to the closing state $s$). Then, the auditor, upon receiving the state history from both parties, verifies it by re-creating the hash-chain.
If a party does not respond to the access request or presents a different state history the auditor pursues external punishment.
\caption{\texttt{Audit}}
\label{alg:audit}
\end{protocol}
 
\section{Evaluation of \sys}\label{sec:evaluation}
In this section we evaluate both the cost of consistent broadcast as well as the cost of \sys's on-chain operations. The key questions we want to answer are:
(a) How does the cost of deploying \sys change as we increase the number of \guards (\ie, the security and fault tolerance), and (b) how does the cost of the off-chain part scale now that we need to interact with the \guards every time in order to protect against network attacks.

\paragraph{Solidity Smart Contract}

To evaluate the cost of deployment of \sys on Ethereum, we have implemented the on-chain operations in the form of a smart contract in Solidity\footnote{The source code of the smart contract is released under the open source MIT license and is available anonymously at \url{https://github.com/ndss2020-brick-submission/brick}.}. \com{The repository will be deanonymized for the final publication.} We measured the gas cost of deploying and operating the smart contract on the Ethereum blockchain. Our results are illustrated in Figure~\ref{fig:gas}. Our gas measurements are conducted using prices as of May 2020, namely the fiat price of $1$ ETH = $195.37$ EUR and a gas price of 20 Gwei (for convenience, prices on the diagram are shown in both EUR and Ether). The contract was implemented in Solidity $0.5.16$ and measurements were performed using using the solc 0.6.8 compiler with optimizations enabled, deployed on a local ganache-cli blockchain using truffle and web3. The measurements concern the deployment of the smart contract, the opening of the channel, optimistically closing the channel, and pessimistically closing the channel.
The contract allows the parties to specify the number $n$ of \guards they desire to involve as well as their identities. To aid with the EVM implementation, we opted for the \texttt{secp256k1} elliptic curve signature scheme~\cite{gura2004comparing,courtois2014optimizing}, as the signatures generated by it can be verified in a gas-efficient manner in Solidity using the \texttt{ecrecover} precompiled smart contract~\cite{wood2014ethereum}. Additionally, the signature scheme makes off-chain signatures compatible with on-chain accounts and as such signatures made off-chain can be verified using public keys available on-chain. We measured the gas cost for \guard values of $n = 3$ to $30$. We recommend the value of $n = 13$ as highlighted in the figure, since it is safe, has good performance and aligns the incentives correctly.

Once the contract is deployed on the Ethereum network, Alice funds it first. Subsequently, once Alice's funding transaction is finalized, Bob funds it. Once Bob's funding transaction is finalized, the collateral can be calculated and so the \guards can fund it in any order simultaneously. When all \guards have funded the contract, any of the two parties can open the channel. At any time prior to opening the channel, any party or \guard can withdraw their money, at which point the channel is cancelled and can no longer be opened, but allows the rest of the parties to withdraw as well, in any order. Once the channel is open, the parties can continue exchanging states off-chain. If multiple \sys channels are used, then the cost of smart contract deployment can be amortized over all of them by abstracting the common functionality into a Solidity library. However, the opening and closing costs are recurrent. We remark here that our cost of deployment ($\approx 9$ EUR) is comparable to other state channel smart contracts which perform different operations under different assumptions (\eg, at current prices, the deployment of the $3$ Pisa~\cite{mccorry2019pisa} contracts amounts to $\approx 17$ EUR).

When the parties wish to close the channel optimistically, initially Alice submits a transaction to the smart contract requesting the channel to close. This request contains Alice's claimed closing state (namely, Alice's value at closing time, as Bob's value at closing time can be deduced from this). Once Alice's transaction is confirmed, if Bob is in agreement, he submits a transaction to the smart contract to signal his agreement. The smart contract then returns Alice's and Bob's values as well as \guards' collateral. Care must be taken to check that the sum of Alice's value and Bob's value at the closing state does not exceed the sum of their values at their initial state, so that sufficient funds remain to return the \guards' collateral. If Bob does not agree with Alice's claim, the channel becomes unusable and must be closed pessimistically (Alice can no longer make an optimistic claim on a different state).  The optimistic close operation measures the cumulative gas cost of the two transactions from both parties. The cost is minimal and should be the normal path since parties need not pay closing fees.

Finally, the channel can be closed pessimistically at any time. To do this, the party who wishes to close the channel requests this from the \guards. Each \guard then submits a transaction to the smart contract containing the sequence number they have seen last, together with Alice's and Bob's off-chain signature on it. These can be submitted in any order. The signatures of Alice and Bob on the sequence number are verified on-chain; this incurs the majority of gas cost for the pessimistic close. The party who wishes to close the channel monitors the chain for such claims and remembers any of them that are fraudulent. As soon as $t$ (honest or adversarial) claims have been recorded, either Alice or Bob can send a transaction to the smart contract to close the channel. The transaction is accompanied by the fraud proofs the closing party was able to assemble, namely the latest announcement for each \guard who made a claim on an earlier sequence number. These announcements contain the signature of the \guard on the plaintext which consists of the smart contract address and the sequence number. As smart contracts have unique addresses, including the smart contract address in the plaintext ensures that the same \guard can participate in multiple \sys channels simultaneously using the same public key\footnote{If \sys is deployed on multiple alt-etherea, each \guard must use a different key in each, or the smart contract must be modified to have the \guard sign the \texttt{CHAIN\_ID} together with the address and sequence number.}. Closing the channel releases the funds of the participant parties and slashes any malicious \guards. After the channel has been closed, any honest \guards who wish to redeem their collateral and their corresponding fee can do so by issuing a further transaction to the smart contract.
The pessimistic close operation was measured when no fraud proofs are provided and includes the transaction of each of the $t$ \guards and the final transaction by one of the parties. Additionally, it was assumed that, while the counterparty is unresponsive or malicious, the \guards were responsive and all submitted the same sequence number (hence limiting the need for multiple signature validations). 

\begin{figure}[t]
    \vspace{-0.3cm}
    \centering
    \includegraphics[width=\columnwidth]{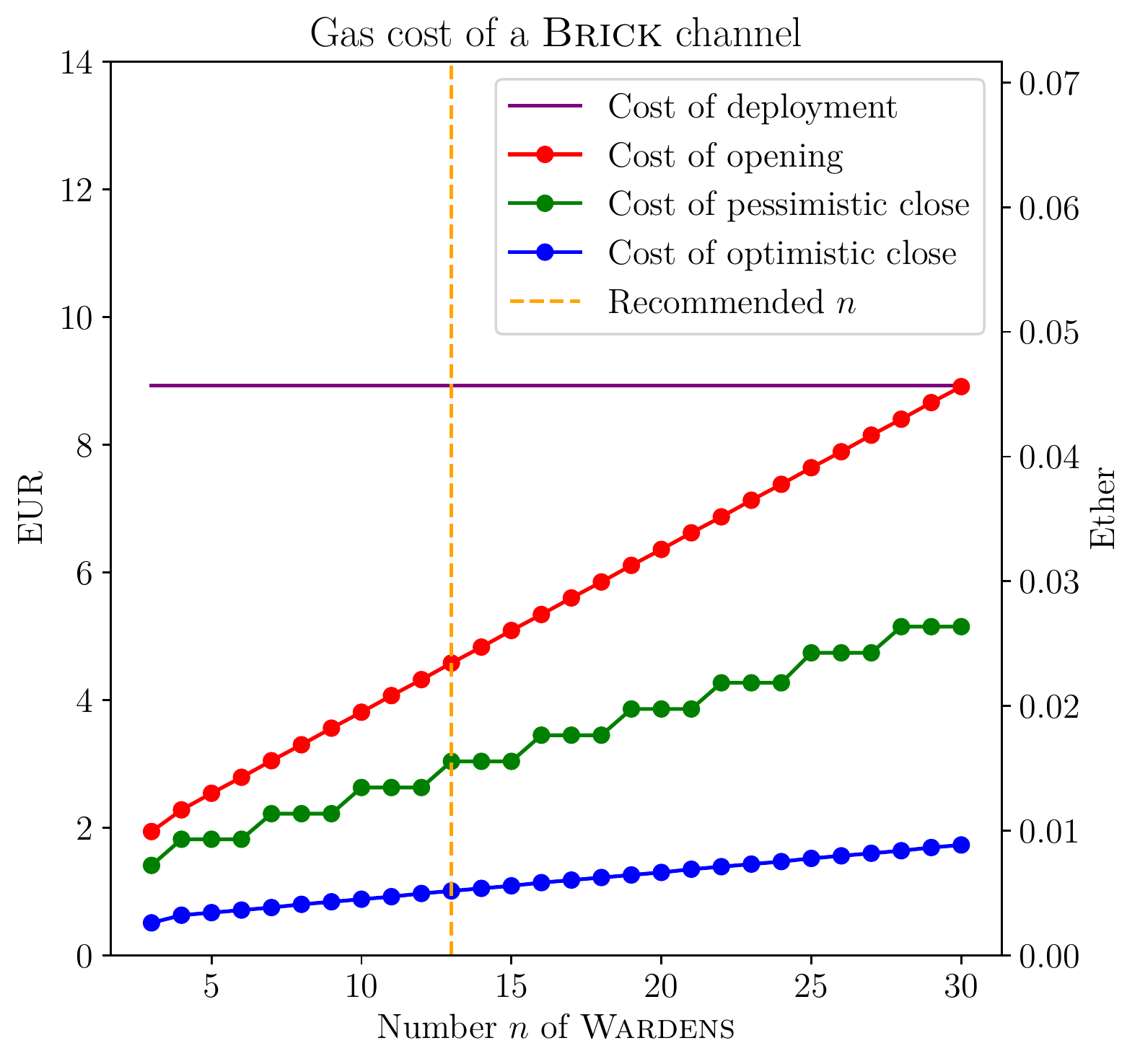}
    \vspace{-0.2cm}
    \caption{The on-chain gas cost of deploying and operating \sys in the form of an Ethereum smart contract.}
    \label{fig:gas}
    \vspace{-0.7cm}
\end{figure}

\paragraph{Consistent Broadcast}
We have also implemented consistent broadcast in Golang using the Kyber~\cite{kyber} cryptographic library and the cothority~\cite{cothority} framework. In Table~\ref{tab:eval}
we evaluate our protocol on Deterlab~\cite{deterlab} using 36 physical machines, each having four Intel E5-2420 v2 CPUs and 24 GB RAM.
To have a realistic wide area network, we impose a 100ms roundtrip latency on the links between \guards and a 35Mbps bandwidth limit.

As illustrated the overhead of using a committee is almost equal to a round-trip latency ($100$ms).
The small overhead is due to the party sending the messages in sequence, hence the last message is sent with a small delay $d>0$. This is observed in the total latency which is close to $100+d$ ms.
This latency in \sys defines the time parties must wait to execute a transaction safely, meaning that \sys provides fast finality. 
These numbers are three orders of magnitude faster than current blockchains, such as Ethereum where blocks are generated on average every 12sec and finality is guaranteed after 6 minutes.
Furthermore, channels are independent and embarrassingly parallel which means that we can deploy as many as we want without significantly increasing the overhead.
In contrast to synchronous channel solutions where finality under our model is guaranteed only after channel closure, \sys provides fast finality (rtt).

\begin{table}[ht]
    \centering
        \caption{Microbenchmark of \sys}

    \begin{tabular}{l|lccc}
        \toprule
                     \textbf{Number of \guards }              & \textbf{7} & \textbf{34} & \textbf{151}\\
        \midrule
                   \textbf{Consistent Broadcast }         &0.1138 sec    &0.118 sec & 0.1338 sec  \\
                   \com{  \textbf{Broadcast}         & &  &  &   \\}
        \bottomrule
    \end{tabular}
    \label{tab:eval}
\end{table}
\com{
As for deployment over existing blockchains, \sys exposes a classic multi-signature wallet~\cite{appleton17ethereum} abstraction.
The funding transaction transfers atomically the collateral of the committee and the state of the parties under the management of the wallet. The closing transaction commits the final state by transferring the funds of the wallet back to its rightful owners. In order to respect the rules of \sys the policy for transferring funds out of the wallet needs to say that either all public-keys of the parties sign (only for \sys) or a combination of $t$-out-of-$n$ \guards together with one party sign. This is a straightforward implementation of a multi-signature wallet on Ethereum. Unfortunately,  it can become rather costly on Bitcoin because of the linear (to the number of parties) possibilities we need to define for the closing transaction. Advancements such as MAST~\cite{rubin14merkelized} can lift this cost if implemented on Bitcoin.
However, we believe that implementing the appropriate incentive mechanisms in Bitcoin might be quite challenging if possible at all. 
}

\section{Discussion, Limitations, and Extensions}\label{sec:limit}
This section briefly outlines several of \sys's important remaining limitations, possible extensions, and discussion on design choices.

\subsubsection*{Byzantine players}
If both channel parties are byzantine then the \guards' collateral can be locked arbitrarily long since the parties can simply crash forever.
This is why in the threat model, we assume that at least the richest channel party is rational to correctly align the incentives.
The second threshold assumption of \sys is that at most $f$ out of the $3f+1$ \guards are byzantine. This threshold is necessary to maintain safety and is in accordance to well known lower bounds for asynchronous consistent broadcast.

\subsubsection*{\guard Unilateral Exit}
As mentioned above, if both parties become malicious they might hold the \guards' collateral hostage. A similar situation is indistinguishable from the parties not transacting often. As a result the \guards might want to exit the channel. A potential extension can support this in two ways.

The first functionality we can incorporate in \sys, is committee replacement. In particular, we can allow a \guard to withdraw its service as long as there is another \guard willing to take its place. In this case, we simply replace the collateral and \guard identities with an update of the funding transaction on-chain, paid by the  \guard that requests to withdraw its service.

Second, if a significant number (e.g, $2f+1$) of \guards declare they want to exit, the smart-contract can release them and convert the channel to a synchronous channel~\cite{mccorry2019pisa}. The parties will now be able to close the channel unilaterally by directly publishing the last valid state.
If the counterparty tries to cheat and publishes an old state, the party (or any remaining watchtower) can catch the dispute on-time and additionally claim the (substantial) closing fee. 

\subsubsection*{Committee Selection}  
Each channel has its own group of \guards, \ie, the committee is independently selected for each channel by the channel parties.
The scalability of the system is not affected by the use of a committee since each channel has its own independent committee of \guards.
The size of the committee for each channel can vary but is constrained by the threat model. If we assume at least one honest party in the channel, a single rational \guard is enough to guarantee the correct operation of \sys. Otherwise, we require more than $7$ \guards to avoid hostage situations from colluding channel parties (see Section~\ref{sec:incentives}). 

The main purpose of introducing a \guard committee is to increase the fault tolerance of the protocol.
The cost of maintaining the channel security for the parties is not dependent on the committee size, but on the value of the channel. If the parties chose a small committee size, the collateral per \guard is high, thus the update fees are few but high. On the other hand, if the parties employ many \guards, the collateral per \guard is low, thus the update fees are many but low. In summary, the cost of security is proportional to the value of the channel and independent of the size of the committee.

\subsubsection*{Consensus vs Consistent Broadcast}
Employing consistent broadcast in a blockchain system typically implies no conflict resolution as there is no liveness guarantee if the sender equivocates. This is not an issue in channels since a valid update needs to be signed by both parties and we provide safety guarantees only to honest and rational parties\footnote{Of course if a party crashes we cannot provide liveness, but safety holds.}.

The state updates in channels are totally ordered by the parties and each sequence number should have a unique corresponding state.
Thereby, it in not the role of the \guard committee to enforce agreement, but merely to verify that agreement was reached, and act as a shared memory for the parties. For this reason, state machine replication (or consensus) protocols are not necessary. Consistent broadcast is tailored for \sys as it offers the only necessary property, equivocation protection.

\subsubsection*{\sys Security under Execution Fork Attacks}
We can extend \sys to run asynchronous consensus~\cite{kokoris19bootstrapping}
during the closing phase in order to defend against execution fork attacks~\cite{karame12double}. 
This would add an one-off overhead during close but would make 
\sys resilient against extreme conditions~\cite{avarikioti2019bitcoin}. 
For example, in case of temporary dishonest majority 
the adversary can attack the persistence\footnote{Persistence states that once a transaction is included in the permanent part of one honest party's chain, then it will be included in every honest party’s blockchain, \ie, the transaction cannot change.} of the underlying blockchain, meaning that the adversary can double-spend funds. The same holds for all channel constructions so far; if the adversary can violate persistence,
the dispute resolution can be reversed, hence funds can be cheated out of a channel party. 

For \sys there is limited protection as the adversary can only close on the last committed state or the freshest valid (but not committed) state. With consensus during close, \sys maintains safety (\ie, no party loses channel funds) even when persistence is violated. A malicious party can only close the channel in the state that the consensus decides to be last, thus a temporary take-over can only affect the channel's liveness. Therefore, \sys can protect both against \textit{liveness and persistence attacks}\footnote{We assume the channel to be created long before these attacks take place, so the adversary cannot fork the transaction that creates the channel.} on the underlying blockchain adding an extra layer of protection, and making it safer to transact on \sys than on the underlying blockchain.

\subsubsection*{Update Fees}
Currently in \sys, the \guards (committee) get rewarded on every state update via a one-way channel. Ideally, these rewards would be included in the state update of the channel. However, even if we modify the update of the state to increase the fees on every update, the parties can always invoke the \texttt{Optimistic Close} protocol, and update the channel state to their favor when closing. Therefore, the incentives mechanism is not robust if the update rewards of the \guards are included in the update states. 

\subsubsection*{Collateral}
The collateral for each \guard is $v/f$, where $v$ is the total value of the channel and $f$ the number of byzantine \guards. This is slightly higher than $v/(f+1)$ that is the lowest amount for which security against bribing attacks is guaranteed when both channel parties and \guards are rational under asynchrony.
Towards a contradiction, we can simply consider a channel where \guards lock  a total collateral $c<v$. Then any rational party that has less than $v-c-\epsilon$ value in the freshest state will bribe the \guards $c+\epsilon$ coins and close the channel in a previous state where the party holds the total value $v$ of the channel.
In a synchronous network, this attack would not work since the other parties would have enough time to dispute. However, under asynchrony (or simple offline parties assumption~\cite{mccorry2019pisa}) there is no such guarantee.
This requirement highlights a trade-off for replacing trust: online participation with synchrony requirements or appropriate incentive mechanisms to compel the honest behavior of rational players.


\subsubsection*{\asys Channel Close} 
In \sys, the parties can close the channel in agreement at any time without the need for committee intervention (\texttt{Optimistic Close}). However, \asys does not allow this function; a channel can only close in collaboration with the committee. Otherwise, the channel parties can update the channel and alter the history without accountability. In case of malicious behavior of the committee (which is outside the threat model of this work), we assume that the parties can alert the auditor who can pursue external punishment to enforce the \guards to be responsive.

\subsubsection*{Decentralization} 
In previous payment channel solutions a party only hires a watchtower if it can count on it in case of an attack. 
In other words, watchtowers in older protocols are the equivalent of insurance companies. If the attack succeeds, the watchtower should reimburse the cheated channel party~\cite{mccorry2019pisa}. After all, it is the watchtower's fault for not checking the blockchain when needed. However, in light of network attacks (which are prevalent in blockchains~\cite{apostolaki16hijacking,gervais15tampering}), only a few, centrally connected miners will be willing to take this risk. 
\sys provides an alternative, that proactively protects from such attacks and we expect to provide better decentralization properties with minimal overhead and fast finality.
\vspace{-0.1cm}
\section{Related Work}
\vspace{-0.1cm}
\begin{table*}[ht]
    \centering
       \vspace{-0.3cm}
        \caption{Comparison with previous work}
        \vspace{-0.3cm}

    \begin{tabular}{l|ccccc}
        \toprule
                     \textbf{Safe under}              & 
                     \textbf{Lightning~\cite{poon2015lightning}} &
                     \textbf{ Monitors~\cite{dryja2016monitors}} &  \textbf{Pisa~\cite{mccorry2019pisa}} &
                     \textbf{Cerberus~\cite{avarikioti2020cerberus}} &
                     \textbf{Brick}\\
        \midrule
                \textbf{Asynchronous Network}     &\xmark & \xmark & \xmark & \xmark & \cmark  \\

        \hline
                \textbf{Offline Parties}     &\xmark & \cmark & $\bm{T \gg t_d}$\footnote{ The party needs to be able to deliver messages and punish the watchtower within a large synchrony bound T.}\label{fn:time} & $\bm{T \gg t_d}$\textsuperscript{18}
                & \cmark \\  
        \hline
                \textbf{Rational Players}         &\xmark & \xmark & $\bm{\sim}$\footnote{ The watchtower needs to lock collateral per-channel, equal to the channel's value. Current implementation of Pisa does not provide this.} & \cmark & \cmark   \\ 
        \hline
                \textbf{Censorship}     &\xmark & \xmark & \xmark & \xmark & \cmark   \\
        \hline
                \textbf{Congestion}     &\xmark & \xmark & \xmark & \xmark & \cmark  \\
        \hline
               \textbf{Forks}     &\xmark & \xmark & \xmark & \xmark & \cmark\footnote{ Possible if consensus is run for closing the channel as described  in Section~\ref{sec:limit}.}  \\
        \hline
                \textbf{Bitcoin Compatible}     &\cmark & \cmark & \xmark & \cmark & \xmark  \\
                
        \bottomrule
    \end{tabular}
            \vspace{-0.3cm}
    \label{tab:comparison}
\end{table*}

In this work we presented \sys, the first incentive-compatible payment channel that provides safety even when parts of the network are disconnected.
We exhibit the differences of \sys to other channel constructions, such as Lightning, as well as other watchtower solutions, such as Monitors, Pisa and Cerberus in Table \ref{tab:comparison}. Specifically, we observe that \sys is the only Layer 2 solution that maintains security under an asynchronous network and offline channel parties while assuming rational \guards. Further, \sys is secure even when the blockchain substrate is censored, and also when the network is congested. Finally, an extension to \sys that we describe  in Section~\ref{sec:limit} enables protection against small scale persistence attacks making it more secure than the underlying blockchain.

Payment channels were originally introduced by Spilman~\cite{spilman2013channels}, as a unidirectional  off-chain solution, meaning that the sender can send to the receiver incremental payments off-chain via the channel, as long as the sender has enough capital on the channel. Later, bidirectional channels were introduced independently by Poon et al.\ with Lightning Network~\cite{poon2015lightning}, and Decker and Wattenhofer with Duplex Micropayment Channels~\cite{DW2015channels}. Bidirectional channels allowed the  capital locked in the channel to be moved in both directions from sender to receiver and back, like in a row of an abacus. All these solutions though require timelocks to guarantee safety and thus make strong synchrony assumptions that sometimes fail in practice.
\sys, on the other hand, is the first bidirectional channel solution that does not require timelocks and operates securely under full asynchrony.


Another safety requirement of the original payment channel proposals is that the participants of a channel are obligated to be frequently online, actively watching the blockchain.
To address this issue, recent proposals introduced third-parties in the channel to act as proxies for the participants of the channel in case of fraud. 
This idea was initially discussed by Dryja~\cite{dryja2016monitors}, who introduced third-parties on Lightning channels, known as Monitors or Watchtowers~\cite{watchtowers}.
Later, Avarikioti et al.\ proposed DCWC~\cite{avarikioti2018towards}, a less centralized distributed protocol for the watchtower service, where every full node can act as a watchtower for multiple channels depending on the network topology. 
In both these works, the watchtowers are paid upon fraud. Hence, the solutions are not incentive-compatible since in case a watchtower is employed no rational party will commit fraud on the channel and thus the watchtowers will never be paid. This means there will be no third parties offering such a service unless we assume they are altruistic.

In parallel, McCorry et al.~\cite{mccorry2019pisa} proposed Pisa, a protocol that enables the delegation of  Sprites~\cite{Miller2017sprites} channels' safety to third-parties called Custodians. Although Pisa proposes both rewards and penalties similarly to \sys, it fails to secure the channels against bribing attacks. Particularly, the watchtower's collateral can be double-spent since it is not tied to the channel/party that employed the watchtower. More importantly, similarly to Watchtowers and DCWC, Pisa demands a synchronous network and a perfect blockchain, meaning that transactions must not be censored, to guarantee the safety of channels. 
On the contrary, \sys is provably secure under rational parties and watchtowers and operates under full asynchrony.
Similarly to Watchtowers and Pisa, \sys also maintains privacy from external to the channel parties, \ie, all the off-chain update states are visible only to the parties of the channel (not to the \guards).

Concurrently to this work, Avarikioti et al.\ introduced Cerberus channels~\cite{avarikioti2020cerberus}, a modification of Lightning channels that incorporates rational watchtowers to Bitcoin channels. Although Cerberus channels are incentive-compatible, they still require timelocks. Thus, in contrast to \sys, their security depends on the synchrony assumptions and a perfect blockchain that cannot be censored. Furthermore, Cerberus channels do not guarantee privacy from the watchtowers, as opposed to \sys and \asys, where unauthorized external parties have no information on the state of the channel.

\subsection*{State Channels, Payment Networks, and Sidechains}

Payment channels are specifically-tailored solutions that support only payments between users.
To extend this solution to handle smart contracts~\cite{szabo1997formalizing} that allow arbitrary operations and computations, \textit{state channels} were introduced~\cite{Miller2017sprites}.
Recently, multiple state channel constructions have emerged, but they mainly focus on the routing problem of channel networks. 
In particular, Sprites~\cite{Miller2017sprites} improve on the worst-case delay for releasing the collateral of the intermediate nodes on the payment network for multi-hop payments. 
In parallel, Perun~\cite{dziembowski2017perun,dziembowski2019multi} proposes a virtual payment hub, where every party can connect to and hence establish a ``virtual channel'' with any other party connected to the hub. 
In a similar manner, Counterfactual~\cite{coleman2018counterfactual} presents ``meta-channels'', which are generalized state channels (not application-specific) with the same functionality as ``virtual channels''.
However, all these constructions use the same foundations, \ie, the same concept on the operation of two-party channels.
And as the fundamental channel solutions are flawed the whole construction inherits the same problems (synchrony and availability assumptions). 
\sys's design could potentially extend to an asynchronous state channel solution if there existed a valuation function for the states of the contract ( \ie, a mapping of each state to a monetary value for the parties) to correctly align incentives.
In this case, the channel can evolve as long as the parties update the state, while in case of of an uncooperative counterparty the honest party can always pessimistically close the channel at the last agreed state and continue execution on-chain.

Another solution for scaling blockchains is sidechains~\cite{back2014enabling,pow-sidechains,pos-sidechains}. In this solution, the workload of the blockchain (main chain) is transferred in other chains, the sidechains, which are ``pegged'' to the main chain. Although the solution is kindred to channels, it differs significantly in one aspect: in channels,
the states updates are totally ordered and unanimously agreed by the parties thus a consensus process is not necessary.
On the contrary, sidechains must operate a consensus process to agree on the validity of a state update.
\sys lies in the intersection of the two concepts; the states are totally ordered and agreed by the parties, whereas \guards merely remember that agreement was reached at the last state announced.

Finally, an extension to payment channels is payment channel networks (PCN)~\cite{bagaria2019boomerang,moreno2017silentwhispers,roos2017settling,prihodko2016flare}.
The core idea of PCN is that users that do not have a direct channel can route payments using the channels of other users. 
While \sys presents a novel channel construction that is safe under asynchrony, enabling asynchronous multi-hop payments remains an open question.

\section{Acknowledgements}

We would like to thank Jakub Sliwinski for his valuable discussion and feedback as well as his impactful contribution to this work.


%

\bibliographystyle{IEEEtranS}
\bibliography{net,sec,os,references}

\appendix

\subsection{Attack against existing channel constructions}\label{app:attack}

Here we present a concrete attack against Pisa~\cite{mccorry2019pisa}. 
This attack applies to any channel construction that has a time-out dispute process.
The attack is simple and relies on delaying the appearance of the dispute transaction on-chain until the dispute window expires. 
The attack is as follows: 

\begin{enumerate}
    \item Alice and Bob establish a state-channel. As Pisa constructs, both Alice and Bob chose their favorite watchtower and provably assign to them the job of watching the chain.
    \item Alice sends multiple coins to Bob through state changes $S_1,S_2,...,S_n$. 
    \item Alice request the closure of the channel with state $S_1$ where Alice holds more coins that $S_n$.
    \item Bob or Bob's watchtower detects the attack and sends on the network a dispute transaction for $S_n$.
    \item Alice attacks Bob in one of the following ways:
    \begin{enumerate}
        \item Congests the blockchain so that the dispute transaction for $S_n$ does not appear on-chain until after the dispute window expires.
        \item Depending on the blockchain, Alice can censor the dispute transaction for $S_n$ without congesting the network~\cite{miller2013feather}.
        \item If Alice is a classic asynchronous network adversary~\cite{bracha87asynchronous}, Alice can delay the delivery of the dispute transaction for $S_n$ until after the dispute window.
    \end{enumerate}
    \item Since the dispute transaction for $S_n$ appears after the dispute window expires, the safety of the channels is violated. 
\end{enumerate}

In Pisa this attack might cause the watchtower to be punished although the watchtower did not misbehave. In other channels~\cite{Miller2017sprites,watchtowers}, that do not employ or do not collateralize watchtowers, this attack leads to  direct value loss for Bob. 
These attacks do not apply in \sys and \asys since the adversary can only temporarily thwart liveness of the \textit{close} function.

\subsection{Incentivizing Rational Parties}
\label{sec:rationalparties}

In this section, we show that the security of \sys holds as long as the richest party of the channel is rational -  the other party can be byzantine. We argue for every phase separately.

\subsubsection*{Open}
Naturally, if a party is not incentivized to open a  channel then that channel will never be opened. We assume the parties have some business interest to use the blockchain and since transacting on channels is faster and cheaper they will prefer it. Deviating from the protocol at this phase is meaningless.

\subsubsection*{Update}
During the execution of Protocol \ref{alg:update}, any party can deviate from the protocol by not signing the hash of the new state. In this case, the new state will not be valid and thus cannot be committed and will not be executed. No party can increase its profit from such behavior directly (attacking the safety of the channel). 
The same argument holds in case the party signs the hash but not the sequence number.
Moreover, attempting to attack the liveness of the channel is not profitable since the counterparty can always request to close in collaboration with the committee by invoking Protocol \ref{alg:pesclose}: \texttt{Pessimistic Close}. 

Lastly, channel parties can collude and stop updating the channel (liveness attack) in order to enforce a hostage situation on the \guards' collateral. However, the committee size is at least $n > 7$. Thereby, from the pigeonhole principle there is at least one party in the channel that has locked funds at least equal to $\frac{v}{2} > \frac{v}{f}$, \ie, the richest party of the channel. Thus, the richest channel party locks an amount larger than each \guard's collateral which means that this party's cost is larger than a \guards. Since we assume the richest party is rational, at any time this party is incentivized to close the channel in collaboration with the committee. 
Thus, the richest party (if rational) will not deviate from the honest execution of Protocol~\ref{alg:update}.

During the execution of \texttt{Consistent Broadcast}, a party can deviate in the following ways:
\begin{itemize}
    \item First, a party can choose not to broadcast the announcement to the committee or part of the committee. 
    In this case, the party has signed the new state, which is now a valid state. This state will be considered committed for the counterparty after the execution of Protocol~\ref{alg:broadcast}.
    We show a rational party cannot increase its payoff by not broadcasting the announcement to all \guards. 
    To demonstrate this, we consider two cases; either the new state is beneficial to the party or not.
    If the new state is beneficial to party $A$, then this state is not beneficial for the counterparty (\eg, $B$ payed $A$ for a service). Thus, if the committee has not received the freshest state, party $B$ can ``rightfully'' close the channel in the previous state. Hence, the expected payoff of party $A$ decreases.
    On the other hand, suppose the new state is not beneficial to party $A$ and it chooses not to send the announcement to the committee. 
    Then, either the counterparty will have the state committed or the state will not be executed.
    From the safety property of the channels, party $A$ cannot successfully close the channel in a previous state if the state was committed. Hence, party $A$ does not increase its payoff.
    On the other hand, if party $A$ does not request the signature of a \guard that will later commit fraud, then the party cannot construct a proof-of-fraud to claim the \guard's collateral and therefore the party's payoff decreases.
    \item Second, a party can broadcast different messages to the committee or parts of the committee. 
    During the execution of Protocol \ref{alg:broadcast}, the \guards verify the parties' signatures, thus an invalid message will not be acknowledged from an honest \guard.
    If the messages are valid (both parties' signatures are present), the parties have misbehaved in collaboration. This can lead to a permanent partition of the view of the committee regarding the state history, but at most one of the states can be committed (get the $2f+1$ signatures).
    Thus, this strategy has the same caveats as the previous one, where the party can only lose from following it.
    \item Lastly, the party can choose not to proceed to the state transition. This is outside the scope of the paper and a problem of a different nature (a fair exchange problem). 
\end{itemize}
Overall, a rational party cannot increase its payoff by deviating from Protocol~\ref{alg:update} or Protocol~\ref{alg:broadcast}.

\subsubsection*{Close}
In this phase, there are two different options:  \texttt{Optimistic Close} and \texttt{Pessimistic Close}.\\
In the first case, where Protocol~\ref{alg:optclose} is executed, a party can deviate from it in the following ways: 
\begin{itemize}
    \item It is the party requesting the closing of the channel in a cheating state. The counterparty will not sign the state since it is being cheated, else it would not be a cheating state. Thus, safety is guaranteed and the party cannot profit from this strategy. 
    \item It is the party requesting the closing of the channel and never publishes the signed closing state. In line 2 of Protocol \ref{alg:optclose}, the signatures on the state are exchanged between the parties, hence the counterparty will eventually publish the closing state. Note that we assume that the closing party sends its signature first with the closing request.
    \item It is the party that got the closing request and does not sign the state. In this case, the party requesting to close the channel can invoke the \texttt{Pessimistic Close} protocol and close the channel in collaboration with the committee in the freshest committed state.
\end{itemize}
Thus, any potential deviation from Protocol \ref{alg:optclose} cannot increase the payoff of a party executing the protocol.\\
In the second case, where Protocol \ref{alg:pesclose} is executed, a party can deviate from it in the following ways:
\begin{itemize}
    \item The party that requested to close is not responsive, meaning that the party does not publish the closing state. In this case, either the party is the richest of the channel or not. If the party is the richest of the two then the party's cost of not responding is higher than that of the counterparty and any other \guard's cost. Therefore, the party is not the richest of the channel. In this case, the counterparty which wants to close the channel, can use the on-chain closing announcements of the \guards and publish the closing state. 
    In both cases, the party that requested close cannot increase its payoff by not revealing the closing state, but only lose from locking its channel funds for a longer period of time that necessary (since no updates are possible).
    \item The party that requested to close the channel publishes an invalid closing state (\eg, random state, previously valid state, a committed state that is not the freshest). However, the smart contract will not close the channel in this state because some verification will fail. Therefore, the party will simply lose the blockchain fee for the on-chain transaction and will not gain any profit. Note that for a rational party to successfully commit fraud, at least one rational \ref{sec:rationalwatchtowers}.
\end{itemize}

To summarize, if the richest channel party deviates from the honest protocol execution at any phase, it cannot increase its payoff. Therefore, the richest rational party will follow the protocol, hence honest behavior for the richest  party is the dominant strategy.

\subsection{\asys Security Analysis}\label{app:plus}

\begin{figure*}[t]
\vspace{-0.3cm}
    \centering
    \includegraphics[width=0.8\textwidth]{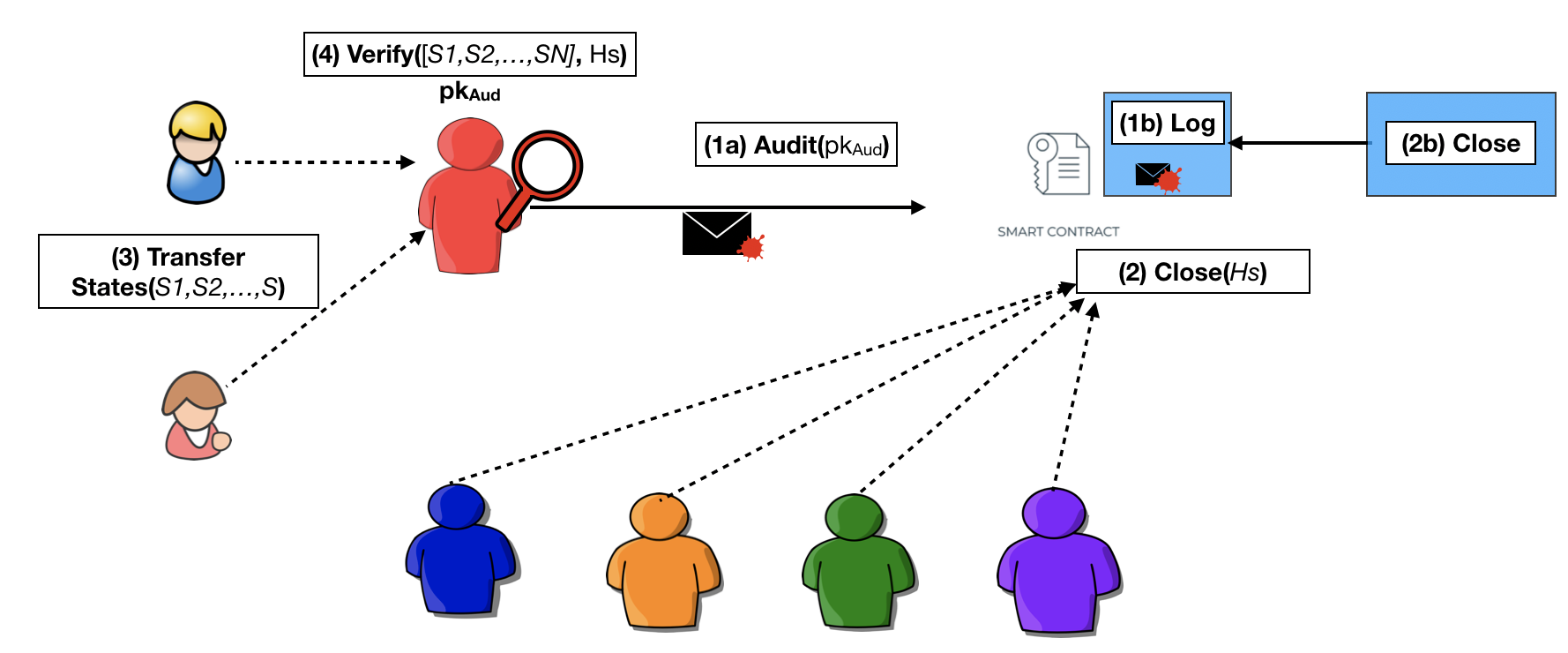}
    \vspace{-0.1cm}
    \caption{ Typical workflow of \asys for an audit update. (1) The auditor starts the audit by posting the request on chain, (2) the committee closes the channel, and (3) the parties transfer the state to the auditor.
    Finally, (4) the auditor cross-checks the claims of the party and the committee.}
    \label{fig:audit}
\end{figure*}


\begin{protocol}
\DontPrintSemicolon
\SetAlgoLined
\textit{Input:} Channel parties $A,B$, current state $s$, previous head of hashchain $H_p$. \;
\textit{Goal:} Create announcement $\{M, \sigma(M), H_s\}$ (sequence number of new state and head of hash chain). \; 
\vspace{0.25cm}
1. Both parties $A,B$ sign and exchange: $\{H_s=H(H_p, H(s_i,r_i), i)\}$, where $r_i$ is a random number and $s_i$ the current state.\;
2. After receiving the signature of the counterparty on $H_s$, the party uses it to create the announcement $\{M, \sigma(M), H_s\} (M=i)$.
\caption{\texttt{Update for \asys}}
\label{alg:plusupdate}
\end{protocol}

In this section, we first prove the \asys goals, namely \textit{Safety, Liveness, Privacy} and \textit{Auditability} assuming $2f+1$ honest \guards, and later we argue that rational \guards will not deviate from the protocol. 

Throughout this section, we assume the richest channel party is rational, thus it will not deviate from honest behavior if it will be discovered and punished. 
Furthermore, we assume the auditor is also rational, meaning that the auditor will only deviate from the protocol to gain more profit (hence the need for a smart contract to do the fair exchange between the auditor and the committee). 
However, the auditor will not punish a party arbitrarily with no proof, since the auditor is supposed to be an external trusted authority (\eg,~judge, regulator, tax office etc.).
Lastly, we assume the threat model of Section~\ref{sec:model} regarding \guards.

\begin{theorem}
{\emph \asys} achieves safety under asynchrony in our system model. 
 \end{theorem}
 \begin{proof}
 It follows from Theorem~\ref{thm:safety}, since all modifications on \sys do not affect the safety property.
 \end{proof}
 
 \begin{theorem}
 {\emph \asys} achieves liveness under asynchrony in our system model.
 \end{theorem}
  \begin{proof}
  It follows from Theorem~\ref{thm:liveness}, since all modifications on \sys do not affect the liveness property.
 \end{proof}
 
 \begin{theorem}
 {\emph \asys} achieves privacy assuming byzantine \guards and parties that want to maintain privacy.
 \end{theorem}
  \begin{proof}
  Suppose an external party learns about the state of the channel during the protocol execution. This means that either the external party intercepted a message (between the parties of the channel or between the parties and the committee) or it is a \guard. In the first case, we assume secure communication channels thus a computationally-bounded adversary cannot get any information from the messages between the parties. In the latter case, the \guards receive during the \textit{Update} phase a message with $H_s=H(H_p, H(s_i,r_i), i)$ for any valid  state update (assuming honest parties that do not intentionally reveal the state). If the \guard extracts the state of the channel $s_i$ from the $H(s_i,r_i)$ (since it knows $H_p$), the \guard reverted the hash function. Therefore, the hash function is not pre-image resistant for a computationally-bounded adversary and hence not cryptographically-secure. This contradicts the system model, where we assumed cryptographically-secure hash functions.
  
  Furthermore, the audit request, which is the first step of the audit functionality, does not leak any information on the channel since the auditor is an external to the channel party. According to Protocol \ref{alg:audit}, the parties publish the closing state on-chain to close the channel (Protocol~\ref{alg:pesclose}). Thus, privacy is preserved since by definition it is only guaranteed until at least one of the parties initiates the closing of the channel.
 \end{proof}
 
 \begin{theorem}
 {\emph \asys} achieves auditability under asynchrony in our system model.
 \end{theorem}
 \begin{proof}
 Since both the \guards and a party of the channel are rational, Protocol~\ref{alg:audit} will not complete (a valid closing state will not be published) unless the request for access is valid, hence the auditor is authorized. Thus, to prove \asys satisfies auditability, it is enough to prove that every committed state is verifiable by a third party.
 
 Towards contradiction, suppose $s$ is the earliest (least fresh) committed state that is not verifiable (since there is at least one).
 This means that state $s$ is replaced by another state $s'$ either in the history of all parties or in the hash-chain that produced the hash-head corresponding to the closing state. 
 If $s$ is not part of the hash-chain, but it is committed, then the parties misbehaved and sent to at least one \guard a different state. Note that the \guard cannot misbehave since the announcement must have the signatures of both parties. Furthermore, two different states cannot be simultaneously committed since this would require two quorums of $2f+1$ \guards that signed different state updates, thus at least $f+1$ Byzantine \guards. Contradiction to our threat model.
 Therefore, state $s$  is not part of the state history provided by the parties.
 In both cases, the auditor punishes the parties (externally).
 And since we assume one party of the channel is rational (and the external punishment exceeds the potential gain of cheating), the party will not misbehave.
 Thus, every committed state is verifiable.
 \end{proof}



\end{document}